\documentclass[a4paper, 11pt]{llncs}
\usepackage{amsmath}
\usepackage{amsfonts}
\usepackage{amssymb}
\usepackage{graphicx}
\usepackage{ulem}
\newtheorem{fact}{Fact}

\providecommand{\U}[1]{\protect\rule{.1in}{.1in}}

\begin{document}
%\abovedisplayskip=0.5\abovedisplayskip
%\belowdisplayskip=0.5\belowdisplayskip
\pagestyle{plain}

\title{\large Representation of Quantum Circuits with Clifford and $\pi/8$ Gates}%

\author{Ken Matsumoto \and Kazuyuki Amano}%

%\keywords{Representation of Unitary Group, Quantum Circuit, Normal form.}

\institute{Department of Computer Science, Graduate School of Engineering, Gunma University\\
Tenjin 1-5-1, Kiryu, Gunma 376-8515 Japan\\
\tt matsumoto@ja4.cs.gunma-u.ac.jp, amano@cs.gunma-u.ac.jp} 

\maketitle

\begin{abstract}
In this paper, we introduce the notion of a normal form of
one qubit quantum circuits over the basis $\{H, P, T\}$,
where $H$, $P$ and $T$ denote the Hadamard, Phase and $\pi/8$
gates, respectively.
This basis is known as the {\it standard set} and its
universality has been shown by Boykin et al. [FOCS '99].
Our normal form has several nice properties: (i)
Every circuit over this basis can easily be transformed
into a normal form, and (ii) Every two normal form
circuits compute same unitary matrix if and only if both circuits
are identical.
We also show that the number of unitary operations that can be 
represented by
a circuit over this basis that contains at most $n$ $T$-gates is 
exactly $192 \cdot (3 \cdot  2^n - 2)$.
\begin{flushleft}
\end{flushleft}

\end{abstract}

\section{Introduction and results}

Quantum computing is a very active area of research because of its ability 
to efficiently solve problems for which no efficient classical algorithms are known. 
For example, it is possible for a quantum computer to 
solve integer factorization in polynomial time with Shor's algorithms \cite{Sho94}. 
However, it is not yet known whether quantum computers are strictly more powerful than classical computers. 

Quantum algorithms are realized by a quantum circuit consisting of basic gates 
corresponding to unitary matrices.
In other words, the design of quantum algorithms can be seen as a
decomposition of a unitary matrix into a product of matrices chosen from a basic set.
A discrete set of quantum gates is called {\it universal}
if any unitary transformation can be approximated with an arbitrary 
precision by a circuit involving those gates only.
For example, Boykin et al. \cite{B99} proved that 
the basis $\left\{ H , T , CNOT \right\}$ is universal, where $H$C$T$, and $CNOT$ are called 
the Hadamard gate, the $\pi/8$ gate, and the controlled-NOT gate, respectively, and given by
\begin{eqnarray*}
H = \frac{1}{\sqrt{2}}
\begin{pmatrix}
 1 &  1 \\
 1 & -1
\end{pmatrix}
 \ , \ \  T  =   
\begin{pmatrix}
 1 & 0 \\
 0 & e^{i\pi/4}
\end{pmatrix}
 \ , \ \  CNOT  = 
\begin{pmatrix}
 1 & 0 & 0 & 0 \\
 0 & 1 & 0 & 0 \\
 0 & 0 & 0 & 1 \\
 0 & 0 & 1 & 0 
\end{pmatrix}.
\end{eqnarray*}
The basis $\left\{ H , T , CNOT \right\}$ is called the
{\it standard set} \cite[pp. 195]{NC00} and plays a fundamental role in the 
theory of quantum computing 
as the classical universal set $\{AND , NOT\}$ plays in the theory of classical computing. 
%Note also that any $2 \times 2$ unitary matrix can be decomposed 
%with given precision as a product of $H$ and $T$. 

The Solovey-Kitaev theorem (see \cite{DN05} or \cite[Appendix 3]{NC00}) says 
that polynomial size quantum circuits 
over this standard set can solve all the problems in ${\bf BQP}$, 
where ${\bf BQP}$ is the class of problems 
that can be solved efficiently by quantum computers.

The situation is dramatically changed
if we replace the $T$-gate by the $T^2$-gate in this basis.
The gate that performs the unitary operation $P=T^2$ is known as the Phase gate.
Quantum circuits over the basis $\{H, P, CNOT\}$ is usually
called {\it stabilizer circuits} or {\it clifford circuits}.
The Gottesman-Knill theorem says that circuits 
over this basis $\{H, P, CNOT\}$ are not more powerful 
than classical computers (see e.g., \cite[Chap. 10.5.4]{NC00}).
A stronger limitation of clifford circuits has also
been derived \cite{AG04,B06}. 
Recently, Buhrman et al. \cite{B06} showed that 
every Boolean function that can be represented by a clifford
circuit is written as the parity of a subset of input variables
or its negation.

These give an insight that the $T$-gate is 
the root of the power of quantum computing. 
It may be natural to expect that
the research on the effect of the $T$-gate may lead
to better understanding of why
a quantum computer can efficiently compute some hard problems.
%It is also interesting to see how many $T$-gates are required 
%to implement $2 \times 2$ unitary matrix for a given precision. 

In this paper, we concentrate on {\it one qubit} circuits over the
standard set, i.e., $\{H, T\}$ and analyze the properties
of them.
It seems difficult to give an efficient representation for a 
given unitary matrix with elements of such a discrete universal set, 
because a relation between a quantum circuit and 
the corresponding unitary matrix is not clear. 
However, if a good representation is found, it will be useful 
for designing an efficient quantum circuit.

The main contribution of this paper is as follows:
We introduce a representation named {\it normal form} for one qubit 
circuits over the universal basis $\{H, T\}$. 
Let ${\cal C}_1$ be the set of $2 \times 2$ unitary matrices
that can be represented by a circuit over the clifford basis $\{H,P\}$.
The set ${\cal C}_1$ forms a group known as {\it Clifford group}
and has order 192.
Our normal form is defined recursively as follows.
\begin{enumerate}
 \item[(a)] For each $D \in {\cal C}_1$, 
 a shortest circuit over $\{H, P\}$ that represents $D$ is a normal form
 (we break ties arbitrarily).
 \item[(b)] If $C$ is a normal form whose leftmost (closest to the output) gate is not $T$, then each of $TC$, $HTC$, and $PHTC$ is a normal form. 
\end{enumerate}
Equivalently, a normal form circuit is of the form
$W_n T W_{n-1} T \cdots T W_1 T W_0$ for some $n \geq 0$ where
$W_n \in \{I, H, PH\}$, $W_{i} \in \{H, PH\}$ for $i=1, \ldots, n-1$
and $W_0 \in {\cal C}_1$.

Our normal form has several good properties :
\begin{enumerate}
 \item[(1)] a normal form circuit has high regularity, 
 \item[(2)] every one qubit circuit over $\{H,T\}$ (or $\{H,P,T\}$)
            can easily be transformed into an equivalent normal form
						circuit, and
 \item[(3)] two normal form circuits perform same computation 
            if and only if both circuits are identical. 
\end{enumerate}
(3) is a surprising property. This enables us to decide
whether two normal form circuits perform same computation without 
calculating the matrix product.  
This can also be used to estimate the number of
$2 \times 2$ unitary matrices represented by a circuit over 
$\left\{H , P , T \right\}$ with at most $n$ $T$-gates. 
The number is exactly $192\cdot (3 \cdot 2^n - 2)$.

This paper is organized as follows. In Section \ref{Preliminaries}, we introduce our definitions and notations. 
In Section \ref{def_normal_form}, we define the normal form and discuss its properties. 
Section \ref{prove_normal_form} is devoted to the proof of our
main result.
In Section \ref{number_of_T_gates}, we discuss the
number of matrices that can be represented by a circuit 
over $\{H,P,T\}$ with a limited number of $T$-gates.

\section{Preliminaries}
\label{Preliminaries}

In this section, we introduce the definitions and notations 
needed to understand the normal form of circuits. 

Throughout the paper, we concentrate on {\it one qubit} quantum
circuits.
A one qubit quantum circuit can be represented by a string 
consisting of symbols each of which represents a gate.
An operation performed by a gate is represented by a unitary
matrix of degree two.
For example, $HPP = HP^{2}$ expresses a circuit which 
performs operations $P$, $P$ and $H$ in this order from the input side. 
By convention, when we draw a circuit, the input
is on the right side and the computation proceeds from right to left (see Figure \ref{fig:circuit_mikata}). 
\begin{figure}[t]
  \begin{center}
    \includegraphics[clip, scale = 0.29]{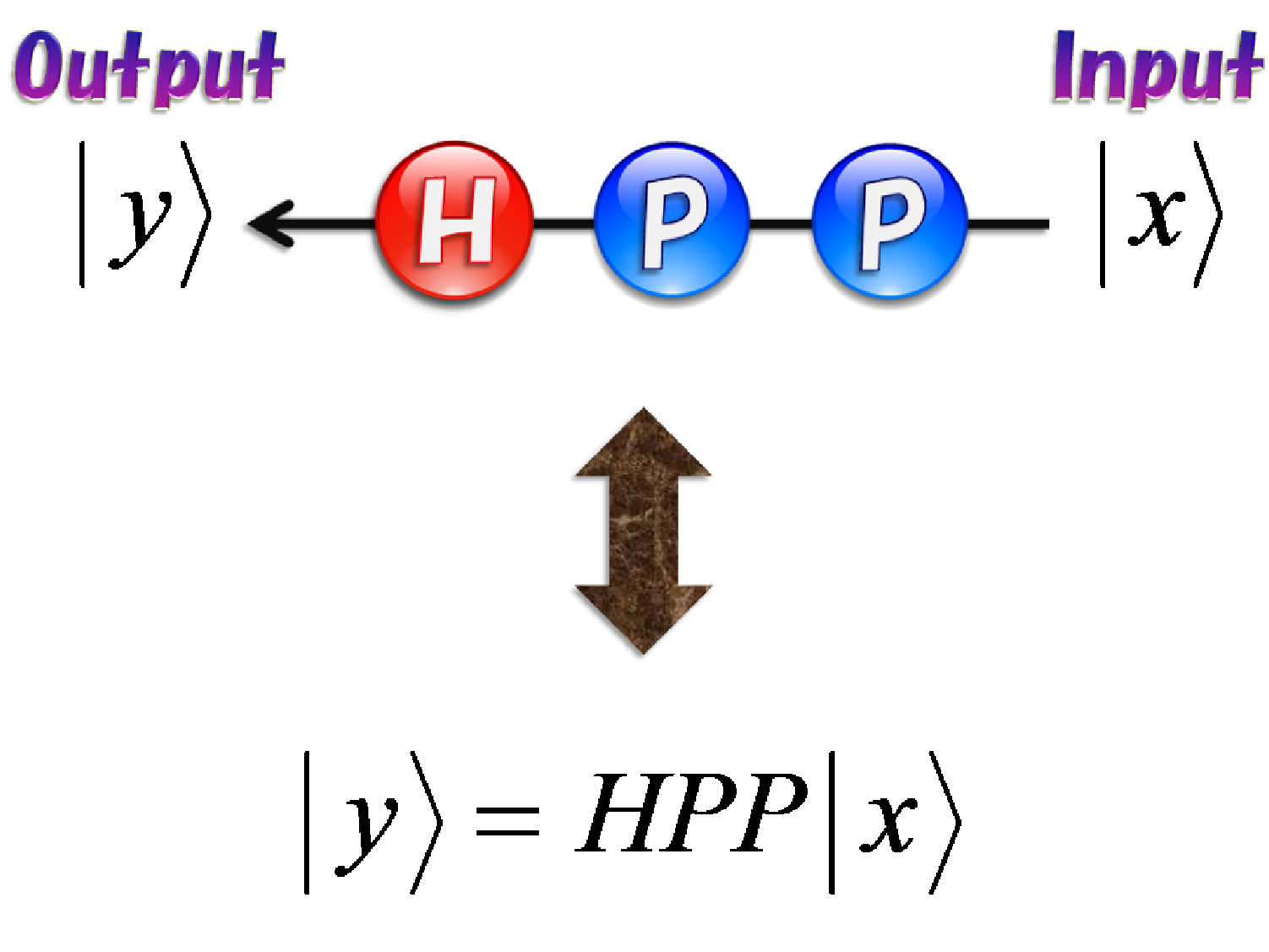}
    \caption{A circuit and the corresponding computation.}
    \label{fig:circuit_mikata}
  \end{center}
 \end{figure}
We usually distinguish a circuit from the matrix computed by
the circuit, 
because different circuits may yield same computation.
For example, two circuits $HHP$ and $PHH$ perform the same computation
since $H\cdot H \cdot P = P \cdot H \cdot H = P$. 
We define the quantum circuit family as follows.
\begin{definition}
\label{def:circuit_family}
Let ${\bf X} = \left\{ G_{1},G_{2},\cdots,G_{m} \right\}$ be the set of symbols, or ``gates". 
%The empty string is allowed, and we denote it by $\varepsilon$. 
The set of all strings on ${\bf X}$ is denoted by $F\left( {\bf X} \right)$. 
A map $f_{u}$ from $F\left( {\bf X} \right)$ to a set of unitary matrices of degree two is defined as follows :
For a gate $G \in {\bf X}$, let $f_u(G)$ be a unitary matrix representing an operation performed by $G$.
For $A_1, \ldots, A_n \in {\bf X}$, 
\begin{eqnarray*}
f_{u}\left( A_{n}A_{n-1} \cdots A_{1} \right) \ = \ f_u(A_{n}) \cdot f_u(A_{n-1}) \cdot \cdots  \cdot f_u(A_{1}).
\end{eqnarray*}
In what follows, we say that a circuit $C$ computes the matrix $f_u(C)$.
A set of all unitary matrices computed by a circuit over
${\bf X}$ is given by
\begin{eqnarray*}
f_{u}\left( F\left( {\bf X} \right) \right) \ := \ \left. \left\{ f_{u}\left( x \right) \ \right| \ x \in F\left( {\bf X} \right) \right\}.
\end{eqnarray*}
The equivalence relation on $F\left( {\bf X} \right)$ is defined to be
\begin{eqnarray*}
a \sim b \stackrel{\mathrm{def}}{\Longleftrightarrow} f_{u}\left(a\right) = f_{u}\left(b\right) \ \ , \ \ a,b \in F\left({\bf X} \right).
\end{eqnarray*}
The quantum circuits family on ${\bf X}$, denoted by $C\left( {\bf X} \right)$, is defined as 
\begin{eqnarray*}
C\left({\bf X}\right)  \ \ := \  \ F\left({\bf X}\right) / \sim   \ \ =   \ \left.  \Big\{ \left[ a \right]   \ \right|  \  a \in F\left({\bf X}\right) \Big\},
\end{eqnarray*}
where 
\begin{eqnarray*}
\left[ a \right] \ := \ \left. \Big\{ b \in F\left({\bf X}\right) \right| a \sim b \Big\}.
\end{eqnarray*}
\end{definition}
Here $\left[ a \right]$ is the equivalence class consisting of 
all circuits that computes the same matrix as $a$. 
Therefore $C\left({\bf X}\right)$ is a family of equivalence
classes of strings, or circuits.
%Note also that the map from $C\left({\bf X}\right)$ to $f_{u}\left( F\left( {\bf X} \right) \right)$ 
%is bijection. 
%\begin{eqnarray*}
%C\left({\bf X}\right) \ \cong  \ f_{u}\left( F\left( {\bf X} \right) \right). 
%\end{eqnarray*}

Throughout the paper, we usually distinguish a circuit from the 
corresponding matrix. 
However, when there is no danger of confusion, we will simply denote
$A$ instead of $f_u(A)$.

In this paper, we mainly consider quantum circuits over
two sets of basis $\{H, P\}$ and
$\{H, P, T\}$, where $H$, $P$ and $T$ are the Hadamard, phase
and $\pi/8$ gates, respectively.

%We define the clifford circuit family and the standard circuit family. 

\begin{definition}
Let ${\bf X_{c}} := \left\{ H , P \right\}$ and we call ${\bf X_{c}}$
the {\it clifford basis}.
The clifford circuit family is defined as $C\left({\bf X_{c}}\right)$, 
where 
\begin{eqnarray*}
f_{u}\left( H \right) = \frac{1}{\sqrt{2}}
\begin{pmatrix}
 1 & 1 \\
 1 & -1
\end{pmatrix}
 \ , \ \  f_{u}\left( P \right)  =   
\begin{pmatrix}
 1 & 0 \\
 0 & i
\end{pmatrix}. 
\end{eqnarray*}
%and $i$ is the imaginary unit. 
\end{definition}
\begin{definition}
Let ${\bf X_{s}} := \left\{ H , P , T \right\}$, and
we call ${\bf X_{s}}$ the {\it standard basis}. 
The standard circuit family is defined as $C\left({\bf X_{s}}\right)$, 
where 
\begin{eqnarray*}
f_{u}\left(T\right) = 
\begin{pmatrix}
 1 & 0 \\
 0 & e^{i\pi/4}
\end{pmatrix}.
\end{eqnarray*} 
%and $e$ is the napier's constant. 
\end{definition}

Readers may wonder why the symbol $P$ appears in our
standard basis ${\bf X_{s}}$, because 
\begin{eqnarray*}
f_{u}\left(TT\right) = f_{u}\left( P \right).
\end{eqnarray*} 
We include it in ${\bf X_{s}}$ in order to make the standard circuit family be strictly 
more powerful than the clifford circuit family. 
The complexity of a given unitary matrix is usually defined as the
minimum number of basic gates needed to compute it.
If we don't include $P$ in ${\bf X_{s}}$, then we need
two gates to compute $P$ on the standard basis whereas
it can be computed by a single gate on the clifford basis.

Unitary groups corresponding to the clifford and standard circuit 
families play a fundamental role in the theory of quantum computing. 

\begin{definition}
The clifford group $\mathcal{C}_{1}$ on one qubit is defined as
\begin{eqnarray*}
\left\langle H,P \right\rangle  &:=&  f_{u}\Big( F\left( {\bf X_{c}} \right) \Big). 
%                                 &=&  \left\langle H,P \right\rangle.
\end{eqnarray*}
\end{definition}

In other words,
${\mathcal{C}_1} = \left\langle H,P \right\rangle$ is the set
of all $2 \times 2$ unitary matrices that can be computed by
a circuit over $\{H,P\}$.
The order of $\mathcal{C}_{1}$ is known to be 192. 
Note also that there is a trivial bijection from the clifford circuit family to $\mathcal{C}_{1}$.

\begin{definition}
The standard group on one qubit is defined as
\begin{eqnarray*}
\left\langle H,P,T \right\rangle  & := &  f_{u}\Big( F\left( {\bf X_{s}} \right) \Big).
\end{eqnarray*}
\end{definition}

In other words, $\langle H,P,T \rangle$ is the set of all $2 \times 2$
unitary matrices that can be computed by a circuit over $\{H,P,T\}$.
It is known that $\langle H, P, T\rangle$ is infinite
group and is universal \cite{B99} in a sense that
any $2 \times 2$ unitary matrix can be approximated by a matrix in this group
with an arbitrary precision.
%Note also that the map from the standard circuit family to 
%$\left\langle H,P,T \right\rangle$ is bijection. 

\section{Representative of the standard circuit family}
\label{def_normal_form}
%The following two operations are equivalent:
%The following are equivalent: 
%\begin{enumerate}
% \item[(a)] design of a quantum circuit, 
% \item[(b)] selection of representative from circuit family. 
%\end{enumerate}
In this section, we introduce the notion of a normal form circuit,
which can be used as a representative of classes of the standard 
circuit family.

%\subsection{Representative of the clifford circuit family}

We first define the representative of the clifford circuit family 
as follows. 

\begin{definition}
Let $\left[ a \right] \in C\left({\bf X_{c}}\right)$. 
A representative of
$\left[ a \right]$ is defined to be a shortest string in $\left[ a \right]$ (we break ties arbitrarily). 
A representative of the clifford circuit family
is called {\it a clifford circuit}.
\end{definition}

This seems to be a natural definition, and in fact,
the way of selection will not largely affect the analysis of 
$\mathcal{C}_{1}$, because the order of $\mathcal{C}_{1}$ is
relatively small, say 192.

%\subsection{The normal form}

It seems to be difficult to find representatives of the standard 
circuit family because it is an infinite group. 
Hence we first generate
all the circuits consisting of relatively small number of gates by
using a computer, and 
then analyze them in order to find a ``pattern".
This leads to the following definition of our ``normal form".

\begin{definition}
A normal form circuit is a circuit 
over $\{H, P, T\}$ and is defined recursively as follows :
\begin{enumerate}
 \item[(a)] Every clifford circuit is a normal form. 
 \item[(b)] If $C$ is a normal form whose leftmost gate (closest
 to the output) is not $T$, then each of $TC$, $HTC$, and $PHTC$ is a normal 
form. 
\end{enumerate}
\end{definition}

For example, if $D$ is a clifford circuit, then $TD$ and $PHTHTD$ are
normal form whereas $THPHD$ and $TTD$ are not.
Equivalently, a normal form circuit is of the form
$W_n T W_{n-1} T \cdots T W_1 T W_0$ for some $n \geq 0$ where
$W_n \in \{I, H, PH\}$, $W_{i} \in \{H, PH\}$ for $i=1, \ldots, n-1$, 
$W_0$ is a clifford circuit, and $I$ is the $2 \times 2$ identity matrix 
(see Figures \ref{fig:normal_form} and \ref{fig:how_to_see_Fig_nf1}).
The set ${\bf M}_n$ in Figure \ref{fig:normal_form} is defined
as the set of all matrices that can be computed by a circuit 
over $\{H,P,T\}$ that contains at most $n$ $T$-gates. 
Note that ${\bf M}_{0} = \mathcal{C}_{1}$. 
%Readers will know how to see Figure \ref{fig:normal_form} 
%in Figure \ref{fig:how_to_see_Fig_nf1} and Figure \ref{fig:how_to_see_Fig_nf2}. 
%${\bf M_{n}}$ in Figure \ref{fig:how_to_see_Fig_nf1} and Figure \ref{fig:how_to_see_Fig_nf2} 
%is defined as follows. 

\begin{figure}[t]
  \begin{center}
    \includegraphics[clip, scale = 0.29]{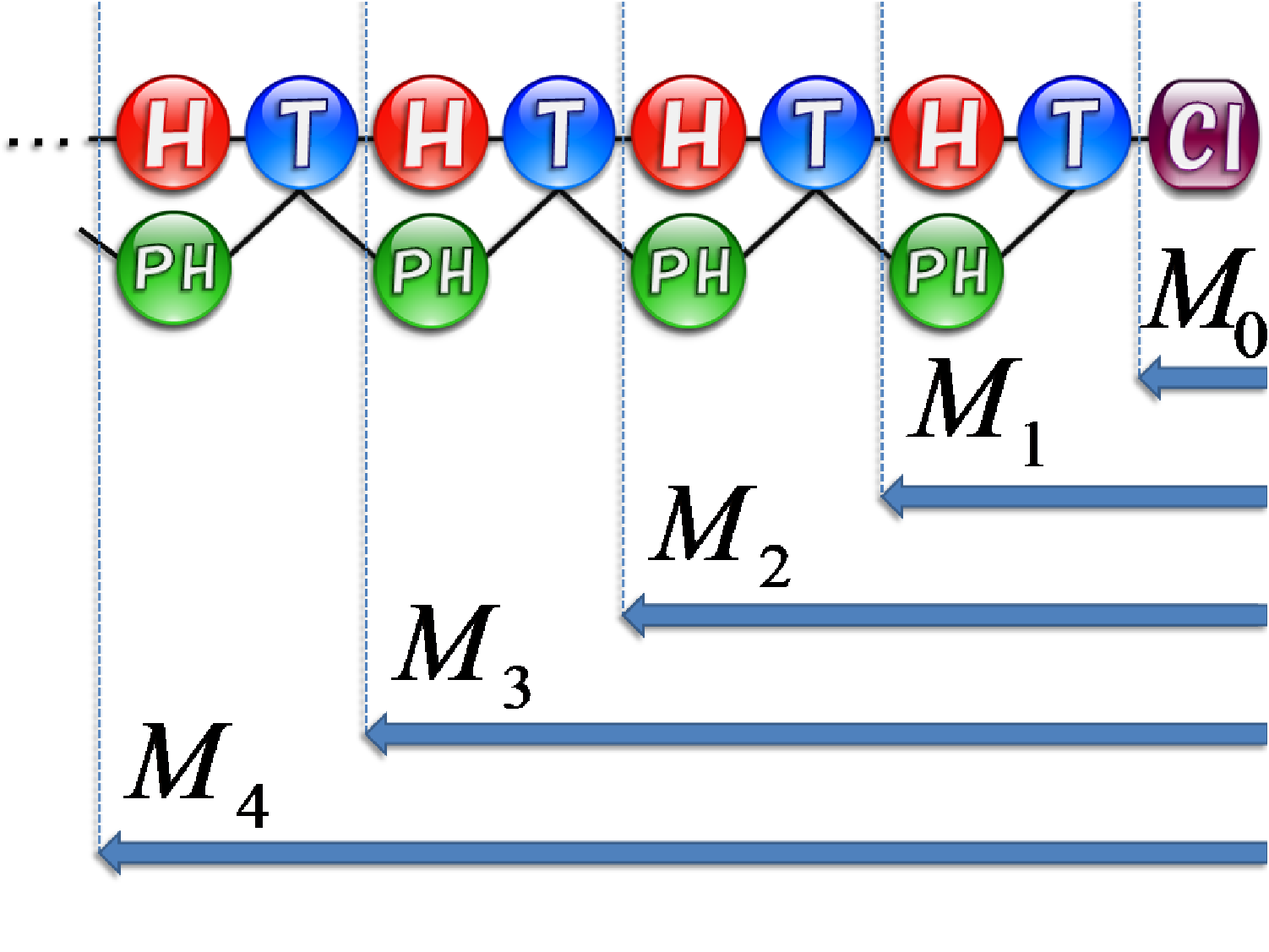}
    \caption{The normal form}
    \label{fig:normal_form}
  \end{center}
 \end{figure}
\begin{figure}[t]
  \begin{center}
    \includegraphics[clip, scale = 0.29]{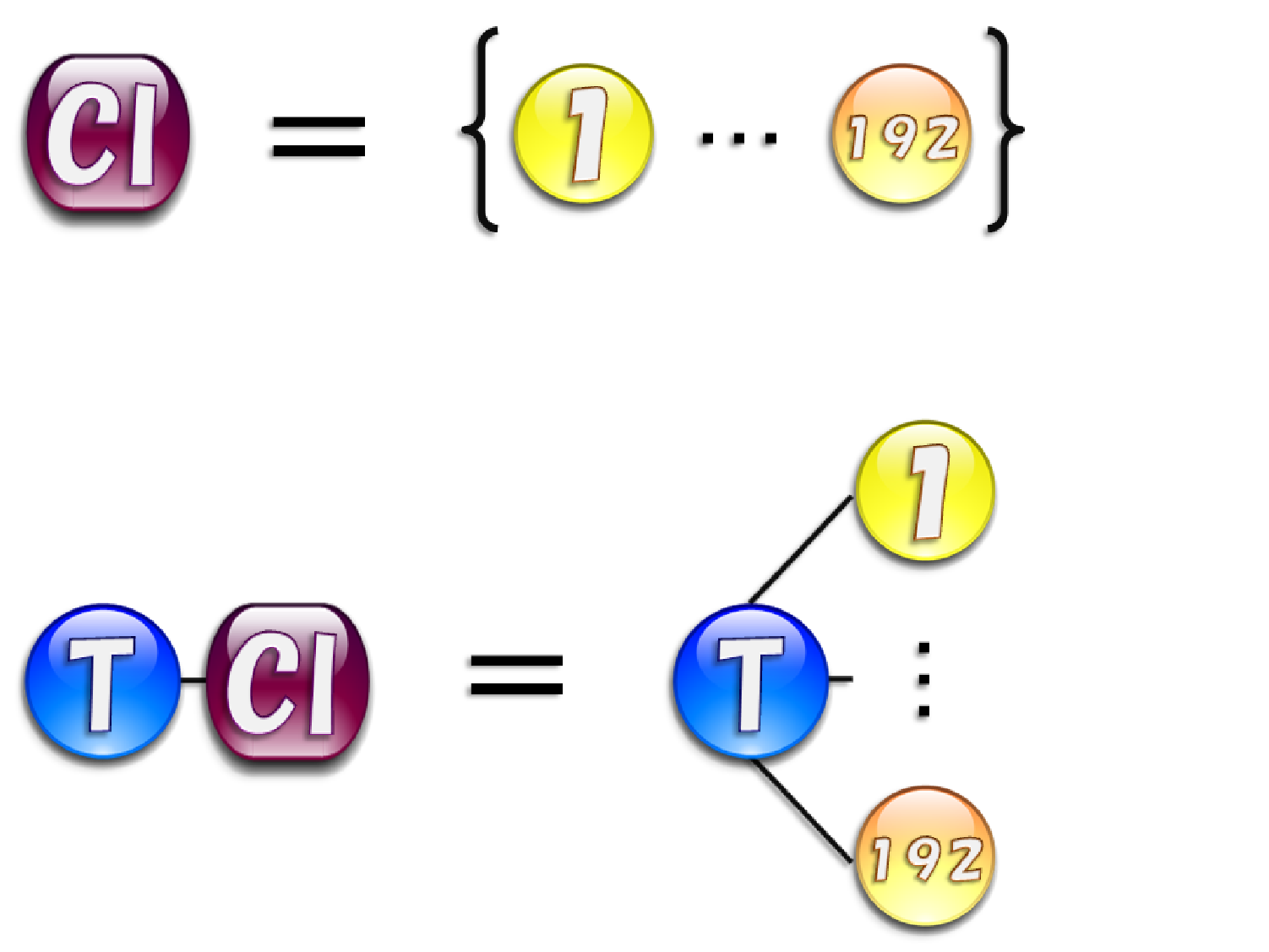}
    \includegraphics[clip, scale = 0.29]{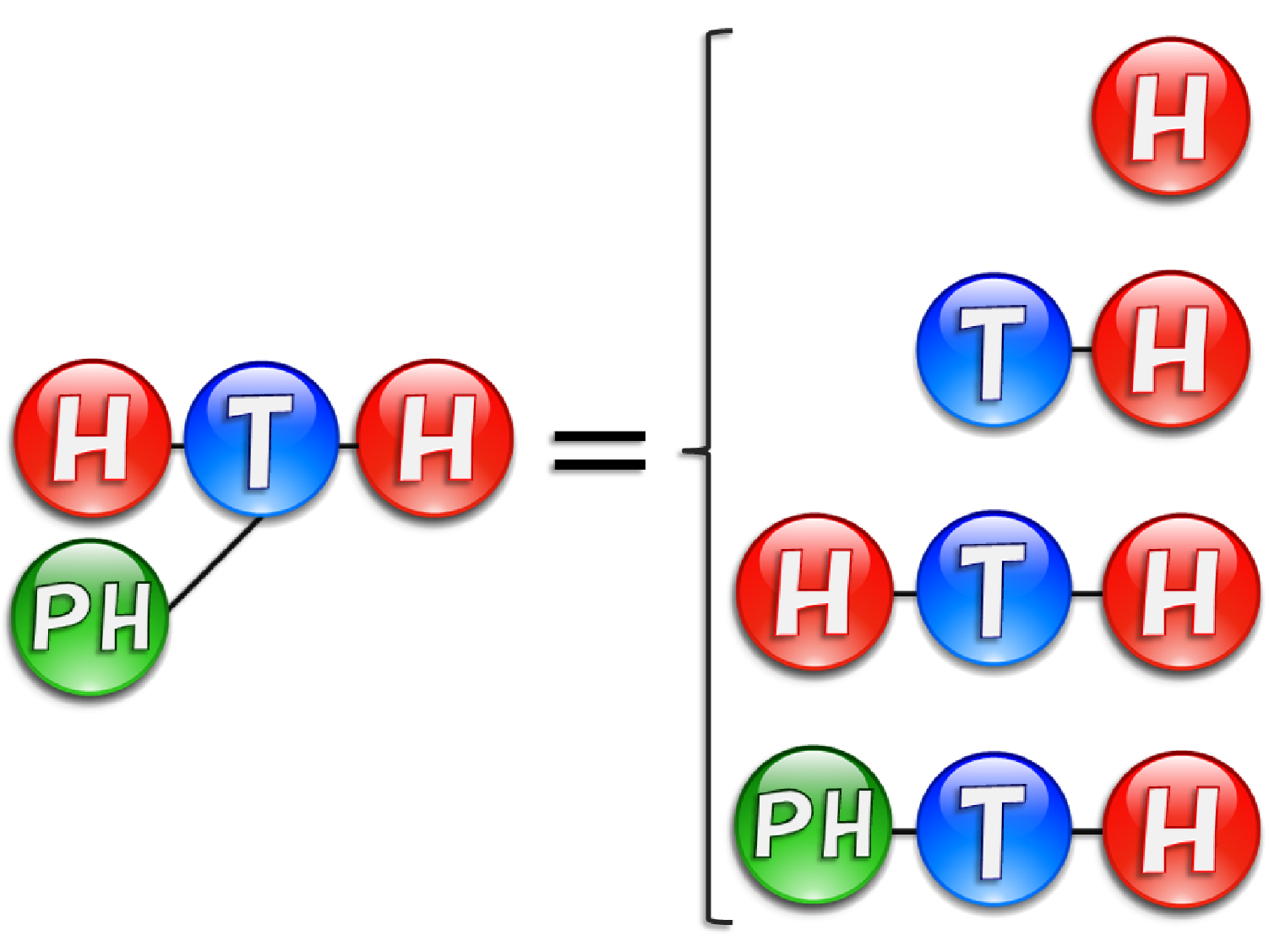}
    \caption{(Left) {\cal Cl} in Figure \ref{fig:normal_form} 
		denotes the set of shortest circuits
		over $\{H,P\}$ for each matrix in ${\cal C}_1$, these are
		denoted by $1 \sim 192$ $(=|{\cal C}_1|)$.
		(Right) A normal form circuit is corresponding to a path from an
		arbitrary chosen gate to one of the rightmost gates in Figure
		\ref{fig:normal_form}.}
    \label{fig:how_to_see_Fig_nf1}
%    \caption{how to Figure \ref{fig:normal_form} (2) : a normal form correspond to a path 
%    from a gate to one of the most right gate in Figure \ref{fig:normal_form}}
%    \label{fig:how_to_see_Fig_nf2}
  \end{center}
 \end{figure}

Our normal form representation is very powerful and appealing, because it has nice properties as follows: 
\begin{enumerate}
\item[(1)] a normal form circuit has high regularity, 
\item[(2)] every one qubit circuit over $\{H,P,T\}$
can easily be transformed into an equivalent normal form circuit,
\item[(3)] two normal form circuits compute same matrix 
if and only if both circuits are identical (comparison can be made as a string). 
\end{enumerate}

\begin{remark}
In this paper, we concentrate on circuits over the basis
$\left\{ H , P , T  \right\}$. 
However, we can also define a normal form for circuits over other bases.
For example, 
for circuits over the basis $\left\{ R , P , T  \right\}$, 
where
\begin{eqnarray*}
R = \frac{1}{\sqrt{2}}
\begin{pmatrix}
 1 &  1 \\
-1 &  1
\end{pmatrix}, 
\end{eqnarray*}
we can show that if we replace $H$ with $R$ in the definition
of our normal form, then the modified normal form 
satisfies all the above properties.
\end{remark}

We now state our main theorem.
\begin{theorem}
\label{theo:normal_form}
The number of normal forms in each equivalence class of $C\left({\bf X_{s}}\right)$ 
is exactly one. 
\qed
\end{theorem}
We prove this theorem in the next section. 
Theorem \ref{theo:normal_form} can also be used to derive the
order of ${\bf M_{n}}$, which will be described in Section \ref{number_of_T_gates}. 

\section{The proof of Theorem \ref{theo:normal_form}}
\label{prove_normal_form}

We divide the proof of Theorem \ref{theo:normal_form} into
two parts :
\begin{enumerate}
 \item[(I)] Every quantum circuit over $\{H,P,T\}$ 
            can be transformed into an equivalent normal form circuit
            (we call this operation ``normalization"). 
 \item[(II)] For every two distinct normal form circuits $C_{1}$ and $C_{2}$, $f_{u}\left( C_{1} \right) \ne  f_{u}\left( C_{2} \right)$ holds.
\end{enumerate}

The statement (I) guarantees that each equivalence class contains at least one
normal form circuit, and the statement (II) asserts the uniqueness.

\subsection{The proof of (I)}
\label{sec:proof_of_I}

%Throughout the proof, we identify a circuit with 
%the corresponding matrix. 

In order to show Statement (I),
we describe the normalization procedure.
We first give a useful property of the clifford group ${\cal C}_1$.
Let $C_{T}({\cal C}_1)$ be a subgroup of ${\cal C}_1$ defined as
 \begin{eqnarray*}
 C_{T}\left( \mathcal{C}_{1} \right) &:=& \left. \left\{ TgT^{-1} \right| g, TgT^{-1} \in \mathcal{C}_{1}  \right\}.
 \end{eqnarray*}
Note that $g \in C_{T}({\cal C}_1)$  if and only if $TgT^{-1} \in C_{T}({\cal C}_1)$, and that a generating set of $C_{T}\left( \mathcal{C}_{1} \right)$ is $\left\{ P , HP^{2}H , \left( HP\right)^{3}  \right\}$, namely
\begin{eqnarray*}
\left\{          
\begin{pmatrix}
 1 & 0 \\
 0 & i
\end{pmatrix}, 
\begin{pmatrix}
 0 & 1 \\
 1 & 0
\end{pmatrix}, 
e^{\frac{\pi}{8}i} \cdot I
\right\}. 
\end{eqnarray*}
Note also that there exists an isomorphism mapping $C_{T}\left( \mathcal{C}_{1} \right) / \mathcal{K}$ into $D_{4}$, namely 
\begin{eqnarray*}
C_{T}\left( \mathcal{C}_{1} \right) / \mathcal{K} \simeq D_{4}, 
\end{eqnarray*}
where 
\begin{eqnarray*}
\mathcal{K} =  \left\{ e^{\frac{i\pi}{8} k} \cdot I
\Bigm| k \in \{0, 1, \ldots, 7\}  \right\}
\end{eqnarray*}
and $D_{4}$ is the dihedral group of degree four. 
The following fact is easily verified by a direct calculation.
\begin{fact}
\label{subgroup_clifford}
The clifford group ${\mathcal C}_1$ can be represented
as follows:
\begin{eqnarray*}
\mathcal{C}_{1} = C_{T}\left( \mathcal{C}_{1} \right) + HC_{T}\left( \mathcal{C}_{1} \right) + PHC_{T}\left( \mathcal{C}_{1} \right), 
\end{eqnarray*}
where 
$HC_{T}\left( \mathcal{C}_{1} \right) := \left. \left\{ Hh \right| h \in C_{T}\left( \mathcal{C}_{1} \right)  \right\},$ and
$PHC_{T}\left( \mathcal{C}_{1} \right)  :=  \left. \left\{ PHh \right| h \in C_{T}\left( \mathcal{C}_{1} \right)  \right\}$. 
\qed
\end{fact}

Note that $C_T({\cal C}_1)$, $HC_T({\cal C}_1)$ and
$PHC_T({\cal C}_1)$ in the above fact are the residue classes of ${\cal C}_1$ and have order 64.
Fact \ref{subgroup_clifford}
guarantees that for every $W_0$ in ${\cal C}_1$,
$W_0 T = S_0 T W_1$ for some $S_0 \in \{I, H, PH\}$
and some clifford circuit $W_1$. 
This gives a set of basic transformation rules of our normalization.
The correctness follows from the definition of $C_T({\cal C}_1)$:
(i) if $W_0 \in C_T({\cal C}_1)$ then $W_0 = T W_1 T^{-1}$ for some
$W_1 \in {\cal C}_1$,
(ii) if $W_0 \in HC_T({\cal C}_1)$ then $W_0 = H T W_1 T^{-1}$ for
some $W_1 \in {\cal C}_1$, and 
(iii) if $W_0 \in PHC_T({\cal C}_1)$ then $W_0 = PH T W_1 T^{-1}$
for some $W_1 \in {\cal C}_1$.

A complete table of our basic transformation rules, which
contains $|{\cal C}_1| = 192$ rules and
categorized into three groups, is given in Appendix.

Roughly speaking, the normalization is
to apply this transformation rule to a given circuit from 
left to right. 
For example, when the initial circuit is given by
$W_3 T W_2 T W_1 T W_0$, where $W_i$, $i=0,1,2,3$ is a clifford 
circuit, the normalization proceeds as described below.
Here we use symbols $W_i$ to denote a clifford circuit, and
$S_i$ to denote a circuit in $\{I, H, PH\}$.
\begin{eqnarray*}
\begin{array}{lcll}
W_{3}  \ T  \ W_{2} \ T \ W_{1} \ T \ W_{0} &=& \ \ \underline{ W_{3}  \ T } \ W_{2} \ T \ W_{1} \ T \ W_{0} & \\
&=& \ \ \displaystyle {\uwave{ S_{0}  \ T \ W_{4} }} \ W_{2} \ T \ W_{1} \ T \ W_{0}  \mbox{\quad}
&\mbox{(apply the rule to } W_3 T) \\
&=& \ \ S_{0}  \ T \ W_{5} \ T \ W_{1} \ T \ W_{0} & (W_5 := W_4 W_2) \\
                                            \ \ &=& \ \ S_{0}  \ T \ \underline{ W_{5} \ T } \ W_{1} \ T \ W_{0} \\
                                            \ \ &=& \ \ S_{0}  \ T \ \uwave{ S_{1} \ T \ W_{6} }\ W_{1} \ T \ W_{0} 
&\mbox{(apply the rule to } W_5 T) \\
                                            \ \ &=& \ \ S_{0}  \ T \ S_{1} \ T \ W_{7} \ T \ W_{0}  & (W_7 := W_6 W_1) \\
                                            \ \ &=& \ \ S_{0}  \ T \ S_{1} \ T \ \underline{ W_{7} \ T } \ W_{0} &  \\
                                            \ \ &=& \ \ S_{0}  \ T \ S_{1} \ T \ \uwave{ S_{2} \ T \ W_{8}} \ W_{0} 
&\mbox{(apply the rule to } W_7 T) \\
                                            \ \ &=& \ \ S_{0}  \ T \ S_{1} \ T \ S_{2} \ T \ W_{9} & (W_9 := W_8 W_0).
\end{array}
\end{eqnarray*}
If $S_i = I$ at some step of the process, then we
``merge" two $T$-gates into a $P$ gate by applying the identity $P=TT$, and
continue the process.
For example, if $S_2=I$ in the above equation,
then we further transform the last circuit to 
\begin{eqnarray*}
S_0 \ T \ S_1 \ T \ S_2 \ T \ W_9 = 
S_0 \ T \ S_1 \ P \ W_9 = S_0 \ T \ W_{10}.
\end{eqnarray*}
It is obvious that, for every input
circuit over $\{H,P,T\}$, the resulting circuit of this normalization
process is a normal form circuit.
This established Statement (I).
\qed

\medskip
It should be noted that the normalization can be performed
in time linear in the number of gates in an initial circuit.

\subsection{The proof of (II)}

The proof of Statement (II) is divided into two subproofs:
\begin{itemize}
 \item[(II-A)] If (II) is false, i.e.,
 there are two distinct normal form circuits $C_1$ and $C_2$ with
 $f_u(C_1) =f_u(C_2)$, then there is a normal form circuit $C$ containing 
 one or more $T$-gates 
                such that $f_{u}\left( C \right) = I$. 
 \item[(II-B)] For every normal form circuit $C$ containing one or more $T$-gates, 
               $f_{u}\left( C \right) \ne I$ holds. 
\end{itemize}

The meaning of the symbols appeared in the proof are as follows. 
\begin{itemize}
 \item $W_a, W_b, \ldots, W_z$, and $W_{j}$ ($j = 0,1,2,\cdots$) denote a clifford circuit, 
 \item $A_{j} ,  B_{j} , C_{j} , D_{j}  \ , \ j = 1,2,\cdots$ denote $H$ or $PH$, 
 \item $A^{\prime} ,  B^{\prime} , C^{\prime} , D^{\prime}  \ , \ j = 0,1,2,\cdots$ 
       denote $I$ or $H$ or $PH$. 
\end{itemize}

%\subsubsection{The lemma for inverse of normal form} \ \\ \ \\ 
Before we proceed to the proof of Statement (II),
we give the following simple lemma that says 
the set ${\bf M_{=n}}$ is closed under the inverse operation, 
where ${\bf M_{=n}} := {\bf M_{n}} \backslash {\bf M_{n-1}}$. 
Recall that ${\bf M_n}$ denotes the set of all matrices
that can be computed by a circuit over $\{H,P,T\}$ with at most
$n$ $T$-gates.
\begin{lemma}
\label{lemma_II-A}
Let $C$ be the normal form circuit containing $n$ $T$-gates. 
If $f_u(C) \in {\bf M_{=n}}$, then $f_u(C)^{-1} \in {\bf M_{=n}}$. 
\end{lemma}
\begin{proof}
It is trivial for $n = 0$. 
Let $n \geq 1$ and suppose that $C$ satisfies 
$f_u(C) \in {\bf M_{=n}}$. 
Then $C$ can be written as 
\begin{eqnarray}
\label{hodai:k+1noseikikei}
A^{\prime} \ T \  A_{n-1} \ T \ \cdots \ T \ A_{2} \ T \ A_{1} \ T  \ W_{a}.
\end{eqnarray}
%Note that Eq. (\ref{hodai:k+1noseikikei}) contains exactly $n$ $T$-gates. 
The inverse matrix of $f_u(C)$ is given by 
\begin{eqnarray*}
W_{a}^{-1}  \ T^{-1} \ A_{1}^{-1}  \ T^{-1} \ A_{2}^{-1} \  \cdots \ 
\left(A^{\prime}\right)^{-1},
\end{eqnarray*}
and is represented as
\begin{eqnarray}
\label{hodai:k+1noseikikeinogyakugen}
W_{n} \ T \ W_{n-1} \ T \ \cdots \ T \ W_{1} \ T \ W_{0} 
\end{eqnarray}
since $T^{-1}=T^7=TP^3$.
%Note that $H^{-1} = H$,$P^{-1} = P^{3}$ and $T^{-1} = T^{7} = TP^{3} = P^{3}T$. 
This implies
%Eq. (\ref{hodai:k+1noseikikeinogyakugen}) contains $n$ $T$-gates, i.e.,
$f_u(C)^{-1} \in {\bf M_{n}}$. 

Suppose that the lemma is false; $f_u(C)^{-1} \in {\bf M_{n-1}}$. 
The above argument gives $(f_u(C)^{-1})^{-1}=f_u(C) \in {\bf M_{n-1}}$, 
which contradicts the assumption that $f_u(C) \in {\bf M_{=n}}$. 
This completes the proof of the lemma.
\qed
\end{proof}

\subsubsection{The proof of (II-A)}
Suppose that (II) is false, i.e., 
there are two distinct normal form circuits $U_a$ and $U_b$ such that
$f_u(U_{a}) = f_u(U_{b})$.
Fix an arbitrary such pair $(U_a,U_b)$ that minimizes 
$t( U_a ) + t( U_b )$,
where $t\left( C \right)$ denotes the number of occurrences of $T$ 
in $C$. 
Put $m=t(U_a)$ and $n=t(U_b)$.
Without loss of generality we assume that $m\geq n$.
We write $U_{a}$ and $U_{b}$ as
\begin{eqnarray*}
\label{siki:saiyounoTwomotukairo}
U_{a} &:=& A^{\prime} T  A_{m-1} \cdots  A_{2}  T  A_{1} T W_{a}, \\
U_{b} &:=& B^{\prime} T  B_{n-1} \cdots  B_{2}  T  B_{1} T W_{b},
\end{eqnarray*}
respectively. 
We can also assume that $A^{\prime} \neq B^{\prime}$ since
otherwise a subcircuit of $U_a$ starting at $A_{m-1}$ and
a subcircuit of $U_b$ starting at $B_{n-1}$ compute
same matrix.
%since otherwise two circuits $A_{m-1} \ldots 

For a while we identify a circuit with the corresponding matrix.
Then we can write 
\begin{eqnarray}
\label{syoumei:TnokosuuKsaisyou}
A^{\prime} T  A_{m-1} \cdots  A_{2}  T  A_{1} T W_{a} \ &=& \ B^{\prime} T  B_{n-1} \cdots  B_{2}  T  B_{1} T W_{b}. 
\end{eqnarray}
The inverse matrix of $f_u(U_b)$ is
\begin{eqnarray}
\label{Eq:Ub-1}
W_b^{-1} T^{-1} B_1^{-1} T^{-1} B_2^{-1} \cdots B_n^{-1} T^{-1} (B')^{-1},
\end{eqnarray}
and this can also be represented by a normal
form with $n$ $T$-gates by Lemma \ref{lemma_II-A}, 
which we write as 
\begin{eqnarray}
\label{syoumei:Bnogyakugen_Tonaji}
C^{\prime} T  C_{n - 1} \cdots  C_{2}  T  C_{1} T W_{c}. 
\end{eqnarray}
Here we divide the proof into two cases.
 
\noindent
(Case\ 1) $m > n$. 

By multiplying Eq. (\ref{syoumei:Bnogyakugen_Tonaji}) from the
right to both sides of Eq. (\ref{syoumei:TnokosuuKsaisyou}), we have
\begin{eqnarray}
\label{labelnonamaegaomoitukanai}
A^{\prime} T  A_{m - 1} \cdots  A_{2}  T  A_{1} T \underline{W_{a}  \ C^{\prime} T}  C_{n - 1} T \cdots  T C_{2}  T  C_{1} T W_{c} &=& I. 
\end{eqnarray}
We consider the normalization of the LHS of 
Eq. (\ref{labelnonamaegaomoitukanai}). 
By applying the basic normalization rule to $W_{a}  \ C^{\prime} T$ 
at the underlined part in Eq. (\ref{labelnonamaegaomoitukanai}),
we obtain
\begin{eqnarray}
\label{labelnonamaegaomoitukanai_a}
A^{\prime} T  A_{m - 1} \cdots  A_{2}  T  A_{1} T \  \uwave{S \  T  W_{0}} C_{n - 1} T \cdots  T C_{2}  T  C_{1} T W_{c} &=& I,
\end{eqnarray}
where $S = I$, $H$ or $PH$. 
If $S = H$ or $PH$, 
then the leftmost $T$ never disappeared
during the normalization process, and hence a normalized circuit
for the LHS of Eq. (\ref{labelnonamaegaomoitukanai})
contains at least one $T$ and computes the identity matrix.

We now assume $S = I$. 
Then Eq. (\ref{labelnonamaegaomoitukanai_a}) is written as
\begin{eqnarray*}
A^{\prime} T  A_{m - 1} \cdots  A_{2}  T  \underline{A_{1} \  P  \  W_{0} C_{n - 1} T} \cdots  T C_{2}  T  C_{1} T W_{c} &=& I. 
\end{eqnarray*}
Note that the left and right $T$-gates of $S$ in 
Eq. (\ref{labelnonamaegaomoitukanai_a}) are disappeared by
applying the identity $P = T^{2}$.  
By applying the basic transformation rule again to $A_{1}   P  W_{0} C_{n - 1} T$ at the underlined part in the above equation,
we obtain
\begin{eqnarray}
\label{labelnonamaegaomoitukanai4}
A^{\prime} T  A_{m - 1} \cdots  A_{2}  T  \uwave{ S T \ \cdot\  }\cdots  T C_{2}  T  C_{1} T W_{c} &=& I. 
\end{eqnarray}
If $S = H$ or $PH$, then Statement (II-A) is established by the
same argument as above.
%in Eq. (\ref{n=1nobaainohenkansiki}) exist in normalization process, 
%then Lemma \ref{lemma_II-A} is true for (case\ 1). 
If $S=I$ for every step of the normalization, then the normalized
circuit for the LHS of Eq. (\ref{labelnonamaegaomoitukanai}) has 
$m - n$ $T$-gates.
These complete the proof of Case 1.

\noindent
(Case\ 2) $m=n$. 

For $n = m \leq  1$,  we can check the statement by a direct computation.

Let $n = m>1$. 
By multiplying Eq. (\ref{Eq:Ub-1}) from the
right to both sides of Eq. (\ref{syoumei:TnokosuuKsaisyou}), we have
\begin{eqnarray*}
A^{\prime} T  A_{n - 1} \cdots  A_{2}  T  A_{1} T \underline{W_a W_b^{-1} T} P^3 B_1^{-1} T^{-1} \cdots T^{-1} (B')^{-1} & = &  I. 
\end{eqnarray*}
By applying the basic transformation rule to the underlined part
in the above equation, 
%in Eq. (\ref{labelnonamaegaomoitukanai}), 
we obtain
\begin{eqnarray}
\label{labelnonamaegaomoitukanai2}
A^{\prime} T  A_{n - 1} \cdots  A_{2}  T  A_{1} T \uwave{S T W_z} P^3 B_1^{-1} T^{-1} \cdots T^{-1} (B')^{-1} & = &  I. 
\end{eqnarray}
If $S = H$ or $S = PH$, 
then we can see that
the normalized circuit of the LHS of Eq. (\ref{labelnonamaegaomoitukanai2})
contains at least one $T$ by the same argument to the proof of Case 1.

Assume that $S = I$, i.e., 
\begin{eqnarray*}
\label{labelnonamaegaomoitukanai3}
A^{\prime} T  A_{n - 1} \cdots  A_{2}  T  A_{1} P  W_z P^3 B_1^{-1} T^{-1} \cdots T^{-1} (B')^{-1} & = &  I. 
\end{eqnarray*}
This implies
\begin{eqnarray*}
\label{Eq:fin}
A^{\prime} T  A_{n - 1} \cdots  A_{2}  T  A_{1} = B' T  \cdots T (PW_zP^3B_1^{-1})^{-1}. 
\end{eqnarray*}
If we replace the rightmost term by an equivalent clifford 
circuit, then both sides in the above equation is a normal form circuit
that contains $n-1$ $T$-gates. 
In addition, since $A^{\prime} \neq B^{\prime}$, these are 
different circuits.
This contradicts our choice of $U_a$ and $U_b$.
These complete the proof of Case 2, and so the proof of (II-A).
\qed
\medskip
 
%We now proceed to the proof of Statement (II-B).
\subsubsection{The proof of (II-B)}
\medskip
%\begin{theorem}
%\label{theo_II-B}
%Let $C$ be arbitrary normal form containing one or more $T$-gates. 
%Then $f_u(C) \neq I$D
%\end{theorem}

The idea of the proof is borrowed from the stabilizer
formalism \cite[p.454]{NC00} (or see also \cite{AG04,EEC07}). 
Let $|\psi\rangle$ denote a one qubit state :
\begin{eqnarray*}
|\psi\rangle = \alpha |0\rangle + \beta|1 \rangle, 
\end{eqnarray*}
where $|\alpha|^2+|\beta|^2=1$. 
Let
\begin{eqnarray*}
X=\left(\begin{array}{cc}
0 & 1\\
1 & 0
\end{array}\right), \quad
Y=\left(\begin{array}{cc}
0 & -i\\
i & 0
\end{array}\right), \quad
Z=\left(\begin{array}{cc}
1 & 0\\
0 & -1
\end{array}\right).
\end{eqnarray*}
%Then we define matrix $M_{(x,y,z)}$ as follows. 
\begin{definition}
For a tuple of three real numbers $(x,y,z) \in \mathbb{R}^3$,
the matrix $M_{(x,y,z)}$ is defined as
\begin{eqnarray*}
M_{(x,y,z)} := xX+yY+zZ. 
\end{eqnarray*}
We say that $(x,y,z)$ {\it stabilizes} $|\psi\rangle$ if
\begin{eqnarray*}
M_{(x,y,z)} |\psi \rangle = |\psi \rangle. 
\end{eqnarray*}
\end{definition} 
The following two facts are easily verified. 
\begin{fact}
\label{Fact:Z}
If $(x,y,z)$ stabilizes $|0\rangle$, then
$(x,y,z) = (0,0,1)$.
\end{fact}
\begin{proof}
Since $(x,y,z)$ stabilizes $|0\rangle$, we have 
\begin{eqnarray*}
M_{(x,y,z)}\left(\begin{array}{c} 1 \\ 0 \end{array}\right) & = &
\left(\begin{array}{cc}
z & x-yi \\
x+yi & -z 
\end{array}\right) 
\left(\begin{array}{c} 1 \\ 0 \end{array}\right) 
 =  
\left(\begin{array}{c}
z \\ 
x+yi 
\end{array}\right)
=
\left(\begin{array}{c}
1 \\ 
0 
\end{array}\right).
\end{eqnarray*}
This implies
$(x,y,z) = (0,0,1)$ since $x$, $y$ and $z$ are real.
\qed
\end{proof}

\begin{fact}
\label{Fact:trans}
Suppose that $(x,y,z)$ stabilizes $|\psi\rangle$. 
Then $T|\psi\rangle$ is stabilized by $\frac{1}{\sqrt{2}}\cdot(x-y,x+y,\sqrt{2}z)$ 
, $HT|\psi\rangle$ is stabilized by $\frac{1}{\sqrt{2}}\cdot(\sqrt{2}z, -x-y, x-y)$ 
, and $PHT|\psi\rangle$ is stabilized by $\frac{1}{\sqrt{2}}\cdot(x+y, \sqrt{2}z, x-y)$. 
\end{fact}
\begin{proof}
Let $U$ be an arbitrary unitary matrix of degree two, and suppose that 
matrix $M$ of degree two stabilizes $|\psi\rangle$. 
Since  
\begin{eqnarray*}
UMU^\dagger U|\psi\rangle = UM |\psi\rangle = U |\psi\rangle, 
\end{eqnarray*}
$UMU^\dagger$ stabilizes $U|\phi\rangle$. 
Therefore transitions of each stabilizer matrix is given by 
\begin{eqnarray}
HXH^\dagger = Z, \quad HYH^\dagger = -Y, \quad HZH^\dagger = X, 
\label{Eq:transR} \\
PXP^\dagger = Y, \quad PYP^\dagger = -X, \quad PZP^\dagger = Z, 
\label{Eq:transP} \\
TXT^\dagger = \frac{X+Y}{\sqrt{2}}, \quad 
TYT^\dagger = \frac{Y-X}{\sqrt{2}}, \quad
TZT^\dagger = Z.
\label{Eq:trans} 
\end{eqnarray}

Suppose that $|\psi\rangle$ is stabilized by $(x,y,z)$. 
Eq . (\ref{Eq:trans}) gives the stabilizer matrix of $T|\psi\rangle$:
\begin{eqnarray}
\label{Eq:koteiT}
T(xX+yY+zZ)T^\dagger = 
\frac{(x-y)X}{\sqrt{2}} + \frac{(x+y)Y}{\sqrt{2}} + zZ.
\end{eqnarray}
Eq. (\ref{Eq:transR}) and  Eq. (\ref{Eq:koteiT}) give the stabilizer matrix of $HT|\psi\rangle$: 
\begin{eqnarray}
\label{Eq:koteiHT}
H \left( \frac{(x-y)X}{\sqrt{2}} + \frac{(x+y)Y}{\sqrt{2}} + zZ \right) H^\dagger 
 =  zX - \frac{(x+y)Y}{\sqrt{2}} + \frac{(x-y)Z}{\sqrt{2}}. 
\end{eqnarray}
Eq. (\ref{Eq:transP}) and Eq. (\ref{Eq:koteiHT}) give the stabilizer matrix of $PHT|\psi\rangle$: 
\begin{eqnarray*}
P\left( zX - \frac{(x+y)Y}{\sqrt{2}} + \frac{(x-y)Z}{\sqrt{2}} \right) P^\dagger  
& = & \frac{(x+y)X}{\sqrt{2}} + zY + \frac{(x-y)Z}{\sqrt{2}}.
\end{eqnarray*}
\qed
\end{proof}

\medskip

If the number of $T$ in $C$ is one, then we can see that Statement (II-B) is true by a direct computation.
Hence we only need to consider normal form circuits with at least 
two $T$-gates.

Let $C$ be a normal form circuit containing $k \geq 2 $ $T$-gates: 
\begin{eqnarray*}
C = C_k T C_{k-1} T \cdots T C_{0}, 
\end{eqnarray*}
where $C_{k} \in \{I, H, PH\}$, $C_{i} \in \{H, PH\}$ for $1 \leq i < k$, and $C_0$ is a clifford circuit. 
For $\ell \leq k$, let $C_{(\ell)}$ be 
a subcircuit of $C$ defined as
\begin{eqnarray*}
C_{(\ell)}=C_{\ell}TC_{\ell-1}T \cdots T C_{0}. 
\end{eqnarray*}

In order to show Statement (II-B), it is sufficient to
show that the stabilizer matrix of $C|0\rangle$ is {\it not} $(0,0,1)$
(by Fact \ref{Fact:Z}). 
To see this, we observe the transition of the stabilizer matrices
of $C_{(\ell)} |0 \rangle$ for $\ell= 0, \ldots, k$.

Since $C_{0}$ contains only $H$ and $P$, 
$C_0|0\rangle$ is stabilized by 
$(x,y,z)\in \{(0,0,\pm1), (0,\pm1,0), \linebreak (\pm1,0,0)\}$
(by Eqs. (\ref{Eq:transR}) and (\ref{Eq:transP})). 
From Fact \ref{Fact:trans}, $C_{(\ell)}|0\rangle$ 
is stabilized by a matrix of the form
\begin{eqnarray}
\label{Eq:type}
\frac{1}{\sqrt{2}^\ell}  (x_a + x_b \sqrt{2} , y_a + y_b \sqrt{2}, 
z_a + z_b \sqrt{2}), 
\end{eqnarray}
where $x_a,x_b,y_a,y_b,z_a,z_b \in \mathbb{Z}$. 
From Eq. (\ref{Eq:type}) and Fact \ref{Fact:trans}, 
$TC_{(\ell)}|0\rangle$, $HTC_{(\ell)}|0\rangle$, 
and $PHTC_{(\ell)}|0\rangle$ are stabilized by 
\begin{eqnarray}
\label{Eq:T}
\frac{1}{\sqrt{2}^{\ell+1}} 
\left(
(x_a-y_a) + (x_b - y_b) \sqrt{2}, \right. 
\left. (x_a+y_a) + (x_b+y_b) \sqrt{2},  2z_b + z_a \sqrt{2} \right), \\
\label{Eq:HT}
\frac{1}{\sqrt{2}^{\ell+1}}
\left(
 2z_b + z_a \sqrt{2}, -(x_a+y_a) - (x_b+y_b) \sqrt{2}, \right. 
\left.(x_a-y_a) + (x_b - y_b) \sqrt{2} \right),  \\
\label{Eq:PHT}
\frac{1}{\sqrt{2}^{\ell+1}}
\left(
 (x_a+y_a) + (x_b + y_b) \sqrt{2},  2z_b + z_a \sqrt{2},  \right. 
 \left.
(x_a-y_a) + (x_b - y_b) \sqrt{2} \right) , 
\end{eqnarray}
respectively. 

Consider a circuit $C$ that contains $\ell$ $T$-gates and $C|0\rangle$
is stabilized by a matrix of the form Eq. (\ref{Eq:type}).
We define nine classes of circuits depending on the parities
of $x_a, x_b, y_a, y_b, z_a$ and $z_b$ in Eq. (\ref{Eq:type}).
%We classify stabilizer matrices of $C\left|0 \right\rangle$ 
%into the following nine classes. 
\begin{itemize}
\item[T1:] $x_b$ and $z_b$ are {\bf odd} numbers, and other four are {\bf even} numbers. 
\item[T2:] $y_b$ and $z_b$ are {\bf odd} numbers, and other four are {\bf even} numbers. 
\item[T3:] $x_b$ and $y_b$ are {\bf odd} numbers, and other four are {\bf even} numbers. 
\item[T4:] $x_b$, $y_a$ and $z_a$ are {\bf odd} numbers, and $x_a$, $y_b$ and $z_b$ are {\bf even} numbers. 
\item[T5:] $x_a$, $y_b$ and $z_a$ are {\bf odd} numbers, and $x_b$, $y_a$ and $z_b$ are {\bf even} numbers. 
\item[T6:] $x_a$, $y_a$ and $z_b$ are {\bf odd} numbers, and $x_b$, $y_b$ and $z_a$ are {\bf even} numbers. 
\item[T7:] $x_a$ is {\bf even} number, and the other five are {\bf odd} numbers. 
\item[T8:] $y_a$ is {\bf even} number, and the other five are {\bf odd} numbers. 
\item[T9:] $z_a$ is {\bf even} number, and the other five are {\bf odd} numbers. 
\end{itemize} 

Note that every circuit $C$ with $C|0 \rangle = |0 \rangle$ does 
not belong to every class since $|0 \rangle$ is stabilized by
$(0,0,1)$, i.e., all of $x_a, x_b, y_a$ and $y_b$ must be even.  
We are now ready to finish  the proof of Statement (II-B).

When $k = 2$, we can confirm that $C\left|0 \right\rangle$ 
is not stabilized by $(0,0,1)$ by
computing all patterns directly, and thus $C|0\rangle \neq |0\rangle$. 
%(Recall that if  $C = I$, then $C\left|0 \right\rangle$ is stabilized by $(0,0,1)$). 
We now assume $k \geq  3$. We divide the proof into two cases depending on the stabilizer matrix of $C_0\left|0 \right\rangle$. 

\medskip
\noindent
(Case\ 1) \ $C_0|0\rangle$ is stabilized by $(0,0,\pm1)$.  \quad

We can easily check that $C_{(2)}|0\rangle=C_2TC_1TC_0|0\rangle$ 
is stabilized by $(x,y,z)=1/2 \cdot (0,\pm \sqrt{2}, \pm \sqrt{2})$ or 
$1/2 \cdot (\pm \sqrt{2},0,  \pm \sqrt{2})$ using Fact \ref{Fact:trans}. 
Namely, $C_{(2)}$ belongs to T1 or T2. 
In addition, it is easy to confirm that Eqs. (\ref{Eq:T}),  (\ref{Eq:HT}) and (\ref{Eq:PHT}) 
give the following two facts.
\begin{fact}
\label{Fact:12}
If $C_{(\ell)}$ belongs to T1 or T2, then $HTC_{(\ell)}$ belongs to T2, 
and $PHTC_{(\ell)}$ belongs to T1. 
\qed
\end{fact}
\begin{fact}
\label{Fact:3}
If $C_{(\ell)}$ belongs to T1 or T2, then $TC_{(\ell)}$ belongs to T3. 
\qed
\end{fact}
By Facts \ref{Fact:12} and \ref{Fact:3},  
we can conclude that
$C$ belongs to T1, T2 or T3.
This implies that the stabilizer matrix of $C|0 \rangle$ is not 
$(0,0,1)$, and hence $f_u(C) \ne I$. 

\medskip
\noindent
(Case\ 2) \ $C_0|0\rangle$ is stabilized by $(0,\pm1,0)$ or $(\pm1,0,0)$.

The proof is analogous to the proof of Case 1.

We can easily verify that $C_{(2)}|0\rangle$ is stabilized by 
$(x,y,z)=1/2 \cdot (\pm \sqrt{2},\pm 1, \pm 1)$ or 
$1/2 \cdot (\pm 1, \pm \sqrt{2} ,  \pm 1)$. 
Namely, $C_{(2)}$ belongs to T4 or T5. 
In addition, it is easy to verify that Eqs. (\ref{Eq:T}), (\ref{Eq:HT}) and (\ref{Eq:PHT}) 
give the following fact. 
\begin{fact}
All of the following is true:
\begin{itemize}
\item[(i)]
If $C_{(\ell)}$ belongs to T4 or T5, then $HTC_{(\ell)}$ belongs to T7, 
and $PHTC_{(\ell)}$ belongs to T8. 
\item[(ii)]
If $C_{(\ell)}$ belongs to T7 or T8, then $HTC_{(\ell)}$ belongs to T4, 
and $PHTC_{(\ell)}$ belongs to T5. 
\label{Fact:59_1}
\item[(iii)]
If $C_{(\ell)}$ belongs to T4 or T5, then $TC_{(\ell)}$ belongs to T9. 
\label{Fact:59_2}
\item[(iv)]
If $C_{(\ell)}$ belongs to T7 or T8, then $TC_{(\ell)}$ belongs to T6. 
\end{itemize}
\qed
\end{fact}
By the above fact, we can show that
$C$ belongs to T4, T5, T6, T7, T8 or T9.
This implies that the stabilizer matrix of $C|0\rangle$ is not 
$(0,0,1)$, 
and hence $f_u(C) \ne I$.
This completes the proof of Case 2, and of Statement (II-B).
\qed

\section{The number of normal form circuits}
\label{number_of_T_gates}

Our main theorem can also be used to derive
the number of
$2 \times 2$ matrices computed by a circuit over the 
standard basis $\{H, P, T\}$ using at most $n$ $T$-gates.
%In other word, we give $\left| {\bf M_{n}} \right|$ and 
%$\left| {\bf M_{=n}} \right|$, where $\left| {\bf M_{n}} \right|$ 
%and $\left| {\bf M_{=n}} \right|$ denote the order of 
%${\bf M_{n}}$ and ${\bf M_{=n}}$ respectively. 
\begin{corollary}
\label{coro_T_1}
For all nonnegative integers $n$,
$\left| {\bf M_{n}} \right| = | {\cal C}_1 | \cdot \left(3 \cdot 2^{n} - 2 \right) 
                           = 192\cdot \left( 3 \cdot 2^{n} - 2 \right). 
$
\end{corollary}
\begin{proof}
The definition of the normal form and Theorem \ref{theo:normal_form} gives
\begin{eqnarray*}
\label{nazonosuuretu}
\left| {\bf  M}_{n} \right| = \left\{ \begin{array}{l} \alpha / 2 \ , \ n=0, \\ 
2  \left| {\bf  M}_{n-1} \right| + \alpha \ , \  n > 0, \end{array} \right. 
\end{eqnarray*}
where $\alpha = 384$. 
The corollary follows from this recurrence formula. 
\qed
\end{proof}
\begin{corollary}
\label{coro_T_2}
For all positive integers $n$,
$\left| {\bf M_{=n}} \right| = 576 \cdot 2^{n-1}$.
\end{corollary}
%\begin{proof}
%\begin{eqnarray*}
%\left| {\bf M_{=n}} \right| &=& \left| {\bf M_{n}} \right| - \left| {\bf M_{n-1}} \right| \\
%                            &=& 576 \cdot 2^{n-1}.
%\end{eqnarray*}
%\qed
%\end{proof}

\section{Concluding remarks} 

In this paper, we introduce the notion of a normal form of 
one qubit quantum circuits over the standard basis $\{H, T, P\}$. 
In addition, we prove that the number of $2 \times 2$ unitary matrices 
computed by  a circuit over $\{H, T, P\}$ that contains at most $n$ $T$-gates
is exactly $192 \cdot (3 \cdot 2^n - 2)$.
%We believe that these results open new vista in the theory of quantum computing, 
%and we will expand them into $n$-qubit. 
Obviously, it is a challenging future work to extend
our result to circuits with multiple qubits.
In other words, our next goal is to give ``$n$-qubit normal form". 

%Multiplying both sides of Eq. ().

%\begin{remark}
% For another simple proof, see, for example, \cite[sec.~15]{Simmons1963}
%\end{remark}
%For example, arbitrary unitary transformations are expressed exactly using one qubit and $CNOT$ gates. 
%Any quantum circuit compute matrix with arbitrary error.

\newpage
\section*{Appendix}
\begin{table}
{\small
\begin{eqnarray*}
\mbox{\hspace*{-3mm}}
\begin{array}{rclrcl}
  IT &=& TI &
  PPHPPPHPHPPPT &=& TPPHPPHPPP\\

  PT &=& TP &
  PPPHPPHPPPT &=& TPPHPPPHPH\\

  PPT &=& TPP &
  HPHPPPHPT &=& TPPPHPPHPP\\

  PPPT &=& TPPP &
  HPHPPHPHPPT &=& THPHPPHPHPP\\

  HPPPHPHT &=& THPPH &
  HPPHPPPT &=& THPHPPPHPPP\\

  HPHPHT &=& THPHPH &
  HPPHPPHPPHT &=& THPPHPPHPPH\\

  HPPPHPHPT &=& THPPHP &
  PHPPHPPT &=& THPPPHPHPPP\\

  PHPPPHPHT &=& TPHPPH &
  HPPPHPPPHPT &=& THPPPHPPPHP\\

  HPHPHPT &=& THPHPHP &
  PPHPPHPT &=& TPHPPPHPHPP\\

  HPPPHPHPPT &=& THPPHPP &
  PPPHPPHT &=& TPPHPPPHPHP\\

  PHPPPHPHPT &=& TPHPPHP &
  HPHPPPHPPT &=& TPPPHPPHPPP\\

  PPHPPPHPHT &=& TPPHPPH &
  HPHPPHPHPPPT &=& THPHPPHPHPPP\\

  HPHPHPPT &=& THPHPHPP &
  HPHPPHPPHPHT &=& THPHPPHPPHPH\\

  HPPHT &=& THPHPPPH &
  HPPHPPHPPHPT &=& THPPHPPHPPHP\\

  HPPPHPHPPPT &=& THPPHPPP &
  HPPHPPPHPPHT &=& THPPHPPPHPPH\\

  PHPPHPPPT &=& THPPPHPH &
  HPPPHPPPHPPT &=& THPPPHPPPHPP\\

  PHPPPHPHPPT &=& TPHPPHPP &
  PPHPPHPPT &=& TPHPPPHPHPPP\\

  PPHPPPHPHPT &=& TPPHPPHP &
  PPPHPPHPT &=& TPPHPPPHPHPP\\

  HPHPPPHPPPT &=& TPPPHPPH &
  HPHPPHPPHPHPT &=& THPHPPHPPHPHP\\

  HPHPHPPPT &=& THPHPHPPP &
  HPPHPPHPPHPPT &=& THPPHPPHPPHPP\\

  HPHPPHPHT &=& THPHPPHPH &
  HPPHPPPHPPHPT &=& THPPHPPPHPPHP\\

  HPPHPT &=& THPHPPPHP &
  HPPPHPPPHPPPT &=& THPPPHPPPHPPP\\

  PHPPHT &=& THPPPHPHP &
  PPPHPPHPPT &=& TPPHPPPHPHPPP\\

  PHPPPHPHPPPT &=& TPHPPHPPP &
  HPHPPHPPHPHPPT &=& THPHPPHPPHPHPP\\

  PPHPPHPPPT &=& TPHPPPHPH &
  HPHPPPHPPHPPHT &=& THPHPPPHPPHPPH\\

  PPHPPPHPHPPT &=& TPPHPPHPP &
  HPPHPPHPPHPPPT &=& THPPHPPHPPHPPP\\

  HPHPPPHT &=& TPPPHPPHP &
  HPPHPPPHPPHPPT &=& THPPHPPPHPPHPP\\

  HPHPPHPHPT &=& THPHPPHPHP &
  HPHPPHPPHPHPPPT &=& THPHPPHPPHPHPPP\\

  HPPHPPT &=& THPHPPPHPP &
  HPHPPPHPPHPPHPT &=& THPHPPPHPPHPPHP\\

  PHPPHPT &=& THPPPHPHPP &
  HPPHPPPHPPHPPPT &=& THPPHPPPHPPHPPP\\

  HPPPHPPPHT &=& THPPPHPPPH &
  HPHPPPHPPHPPHPPT &=& THPHPPPHPPHPPHPP\\

  PPHPPHT &=& TPHPPPHPHP & 
  HPHPPPHPPHPPHPPPT &=& THPHPPPHPPHPPHPPP
	\end{array}
  \end{eqnarray*} }
	\caption{The transformation rules of the form $W_0 T = T W_1$.}
	\end{table}
 \newpage
 
 \begin{table}
{\small
\begin{eqnarray*}
\mbox{\hspace*{-5mm}}
\begin{array}{rclrcl}
%1
  HT &=& HT &
  HPPHPPPHPHPPPT &=& HTPPHPPHPPP\\

%2
  HPT &=& HTP &
  HPPPHPPHPPPT &=& HTPPHPPPHPH\\

%3
  HPPT &=& HTPP &
  PHPPPHPT &=& HTPPPHPPHPP\\

%4
  HPPPT &=& HTPPP &
  PHPPHPHPPT &=& HTHPHPPHPHPP\\

%5
  PPPHPHT &=& HTHPPH &
  PPHPPPT &=& HTHPHPPPHPPP\\

%6
  PHPHT &=& HTHPHPH &
  PPHPPHPPHT &=& HTHPPHPPHPPH\\

%7
  PPPHPHPT &=& HTHPPHP &
  HPHPPHPPT &=& HTHPPPHPHPPP\\

%8
  HPHPPPHPHT &=& HTPHPPH &
  PPPHPPPHPT &=& HTHPPPHPPPHP\\

%9
  PHPHPT &=& HTHPHPHP &
  HPPHPPHPT &=& HTPHPPPHPHPP\\

%10
  PPPHPHPPT &=& HTHPPHPP &
  HPPPHPPHT &=& HTPPHPPPHPHP\\

%11
  HPHPPPHPHPT &=& HTPHPPHP &
  PHPPPHPPT &=& HTPPPHPPHPPP\\

%12
  HPPHPPPHPHT &=& HTPPHPPH &
  PHPPHPHPPPT &=& HTHPHPPHPHPPP\\

%13
  PHPHPPT &=& HTHPHPHPP &
  PHPPHPPHPHT &=& HTHPHPPHPPHPH\\

%14
  PPHT &=& HTHPHPPPH &
  PPHPPHPPHPT &=& HTHPPHPPHPPHP\\

%15
  PPPHPHPPPT &=& HTHPPHPPP &
  PPHPPPHPPHT &=& HTHPPHPPPHPPH\\

%16
  HPHPPHPPPT &=& HTHPPPHPH &
  PPPHPPPHPPT &=& HTHPPPHPPPHPP\\

%17
  HPHPPPHPHPPT &=& HTPHPPHPP &
  HPPHPPHPPT &=& HTPHPPPHPHPPP\\

%18
  HPPHPPPHPHPT &=& HTPPHPPHP &
  HPPPHPPHPT &=& HTPPHPPPHPHPP\\

%19
  PHPPPHPPPT &=& HTPPPHPPH &
  PHPPHPPHPHPT &=& HTHPHPPHPPHPHP\\

%20
  PHPHPPPT &=& HTHPHPHPPP &
  PPHPPHPPHPPT &=& HTHPPHPPHPPHPP\\

%21
  PHPPHPHT &=& HTHPHPPHPH &
  PPHPPPHPPHPT &=& HTHPPHPPPHPPHP\\

%22
  PPHPT &=& HTHPHPPPHP &
  PPPHPPPHPPPT &=& HTHPPPHPPPHPPP\\

%23
  HPHPPHT &=& HTHPPPHPHP &
  HPPPHPPHPPT &=& HTPPHPPPHPHPPP\\

%24
  HPHPPPHPHPPPT &=& HTPHPPHPPP &
  PHPPHPPHPHPPT &=& HTHPHPPHPPHPHPP\\

%25
  HPPHPPHPPPT &=& HTPHPPPHPH &
  PHPPPHPPHPPHT &=& HTHPHPPPHPPHPPH\\

%26
  HPPHPPPHPHPPT &=& HTPPHPPHPP &
  PPHPPHPPHPPPT &=& HTHPPHPPHPPHPPP\\

%27
  PHPPPHT &=& HTPPPHPPHP &
  PPHPPPHPPHPPT &=& HTHPPHPPPHPPHPP\\

%28
  PHPPHPHPT &=& HTHPHPPHPHP &
  PHPPHPPHPHPPPT &=& HTHPHPPHPPHPHPPP\\

%29
  PPHPPT &=& HTHPHPPPHPP &
  PHPPPHPPHPPHPT &=& HTHPHPPPHPPHPPHP\\

%30
  HPHPPHPT &=& HTHPPPHPHPP &
  PPHPPPHPPHPPPT &=& HTHPPHPPPHPPHPPP\\

%31
  PPPHPPPHT &=& HTHPPPHPPPH &
  PHPPPHPPHPPHPPT &=& HTHPHPPPHPPHPPHPP\\

%32
  HPPHPPHT &=& HTPHPPPHPHP &
  PHPPPHPPHPPHPPPT &=& HTHPHPPPHPPHPPHPPP\\
  \end{array}
	\end{eqnarray*} } 
\caption{The transformation rules of the form $W_0 T = H T W_1$.
There is a rule $H W_0 T = H T W_1$ in this table
if and only if there is a rule $W_0 T = T W_1$ in the table
for $S_0 = I$ (Table 1).}
\end{table}

\newpage
\begin{table}
{\small
\begin{eqnarray*}
\mbox{\hspace*{-7mm}}
\begin{array}{rclrcl}
%1
  PHT &=& PHT &
  HPPPHPPHPPHPPPT &=& PHTPPHPPHPPP\\

%2
  PHPT &=& PHTP &
  PHPPPHPPHPPPT &=& PHTPPHPPPHPH\\

%3
  PHPPT &=& PHTPP &
  PPHPPPHPT &=& PHTPPPHPPHPP\\

%4
  PHPPPT &=& PHTPPP &
  PPHPPHPHPPT &=& PHTHPHPPHPHPP\\

%5
  HPHT &=& PHTHPPH &
  PPPHPPPT &=& PHTHPHPPPHPPP\\

%6
  PPHPHT &=& PHTHPHPH &
  PPPHPPHPPHT &=& PHTHPPHPPHPPH\\

%7
  HPHPT &=& PHTHPPHP &
  HPPHPHPPPT &=& PHTHPPPHPHPPP\\

%8
  HPPHPPHPHPT &=& PHTPHPPH &
  HPPPHPT &=& PHTHPPPHPPPHP\\

%9
  PPHPHPT &=& PHTHPHPHP &
  PHPPHPPHPT &=& PHTPHPPPHPHPP\\

%10
  HPHPPT &=& PHTHPPHPP &
  PHPPPHPPHT &=& PHTPPHPPPHPHP\\

%11
  HPPHPPHPHPPT &=& PHTPHPPHP &
  PPHPPPHPPT &=& PHTPPPHPPHPPP\\

%12
  HPPPHPPHPPHT &=& PHTPPHPPH &
  PPHPPHPHPPPT &=& PHTHPHPPHPHPPP\\

%13
  PPHPHPPT &=& PHTHPHPHPP &
  HPHPPHPPHPPT &=& PHTHPHPPHPPHPH\\

%14
  PPPHT &=& PHTHPHPPPH &
  PPPHPPHPPHPT &=& PHTHPPHPPHPPHP\\

%15
  HPHPPPT &=& PHTHPPHPPP &
  HPPHPPPHPPPT &=& PHTHPPHPPPHPPH\\

%16
  HPPHPHT &=& PHTHPPPHPH &
  HPPPHPPT &=& PHTHPPPHPPPHPP\\

%17
  HPPHPPHPHPPPT &=& PHTPHPPHPP &
  PHPPHPPHPPT &=& PHTPHPPPHPHPPP\\

%18
  HPPPHPPHPPHPT &=& PHTPPHPPHP &
  PHPPPHPPHPT &=& PHTPPHPPPHPHPP\\

%19
  PPHPPPHPPPT &=& PHTPPPHPPH &
  HPHPPHPPHPPPT &=& PHTHPHPPHPPHPHP\\

%20
  PPHPHPPPT &=& PHTHPHPHPPP &
  PPPHPPHPPHPPT &=& PHTHPPHPPHPPHPP\\

%21
  PPHPPHPHT &=& PHTHPHPPHPH &
  HPPHPPPHT &=& PHTHPPHPPPHPPHP\\

%22
  PPPHPT &=& PHTHPHPPPHP &
  HPPPHPPPT &=& PHTHPPPHPPPHPPP\\

%23
  HPPHPHPT &=& PHTHPPPHPHP &
  PHPPPHPPHPPT &=& PHTPPHPPPHPHPPP\\
	
%24
  HPPHPPHPHT &=& PHTPHPPHPPP &
  HPHPPHPPHT &=& PHTHPHPPHPPHPHPP\\

%25
  PHPPHPPHPPPT &=& PHTPHPPPHPH &
  HPHPPPHPPHPPPT &=& PHTHPHPPPHPPHPPH\\

%26
  HPPPHPPHPPHPPT &=& PHTPPHPPHPP &
  PPPHPPHPPHPPPT &=& PHTHPPHPPHPPHPPP\\

%27
  PPHPPPHT &=& PHTPPPHPPHP &
  HPPHPPPHPT &=& PHTHPPHPPPHPPHPP\\

%28
  PPHPPHPHPT &=& PHTHPHPPHPHP &
  HPHPPHPPHPT &=& PHTHPHPPHPPHPHPPP\\

%29
  PPPHPPT &=& PHTHPHPPPHPP &
  HPHPPPHPPHT &=& PHTHPHPPPHPPHPPHP\\

%30
  HPPHPHPPT &=& PHTHPPPHPHPP &
  HPPHPPPHPPT &=& PHTHPPHPPPHPPHPPP\\

%31
  HPPPHT &=& PHTHPPPHPPPH &
  HPHPPPHPPHPT &=& PHTHPHPPPHPPHPPHPP\\

%32
  PHPPHPPHT &=& PHTPHPPPHPHP &
  HPHPPPHPPHPPT &=& PHTHPHPPPHPPHPPHPPP
	\end{array}
  \end{eqnarray*}  }
\caption{The transformation rules of the form $W_0 T = P H  T W_1$.
There is a rule $PH W_0 T = PH T W_1$ in this table
if and only if there is a rule $W_0 T = T W_1$ in the table
for $S_0 = I$ (Table 1).}
\end{table} 
\end{document}